\documentclass[11pt]{article}

\usepackage[margin=1in]{geometry}
\usepackage{setspace}
\usepackage[T1]{fontenc}
\usepackage[utf8]{inputenc}
\usepackage{lmodern}        % or, for Times-like:
\usepackage{microtype}
\usepackage[round,authoryear]{natbib}

\usepackage{amsmath,amssymb,amsthm,mathtools,bm}

\usepackage{graphicx}
\usepackage{booktabs}
\usepackage{caption}
\usepackage{subcaption}
\usepackage{float}
\usepackage{rotating} % <-- enables \begin{sidewaystable} and \begin{sidewaysfigure}

\usepackage{enumitem}
\usepackage{xcolor}

\usepackage[round,authoryear]{natbib} % author–year citations
\usepackage{hyperref}
\hypersetup{
  colorlinks=true,
  linkcolor=blue,
  citecolor=blue,
  urlcolor=blue
}
\usepackage{cleveref} % load after hyperref

\numberwithin{equation}{section}
\newtheorem{theorem}{Theorem}[section]
\newtheorem{proposition}[theorem]{Proposition}
\newtheorem{lemma}[theorem]{Lemma}

\theoremstyle{definition}

\theoremstyle{remark}

\begin{document}

% --- Title & author block (from IJGT resubmission) ---
\title{Public Communication in Regime Change games\thanks{We are grateful to the editor and the two anonymous referees for their helpful comments. We thank participants of TSE MACRO group and Bordeaux Economics Workshop for valuable feedback. We also thank Evgeny Andreev and Alla Friedman for their feedback on an earlier draft. All remaining errors are our own.}}

\author{Georgy Lukyanov\footnote{Toulouse School of Economics, 1 Esplanade de l'Universit\'{e}, Toulouse, 31000, France. Email: Georgy.Lukyanov@tse-fr.eu} \and Anastasia Makhmudova\footnote{London Business School, The Regent's Park, London NW1 4SA, London, Great Britain.}}

\date{}

\maketitle

\begin{abstract}
We study a regime change game in which the state and an opposition leader both observe the regime’s true strength and may engage in costly communication by manipulating the mean of citizens’ private signals. Each citizen then decides whether to attack the regime; the citizens take their private signals at face value and do not discount or attempt to undo propaganda or counter-propaganda. From the perspective of both the state and the opposition, the size of the attack is uncertain, as the number of committed partisans—those who always attack regardless of their signal—is not observed in advance. We show that a regime on the brink of collapse optimally refrains from propaganda, while the opposition engages in counter-propaganda. The equilibrium level of counter-propaganda increases with the opposition’s benefit-cost ratio and with the precision of citizens’ private signals, and decreases with the cost of attacking.
\end{abstract}



\section{Introduction}

In the information age, propaganda plays an increasingly important role in shaping mass behavior. A striking example is the “Arab Spring,” where social media significantly contributed to the coordination of protests and the eventual collapse of regimes.\footnote{See the discussion in \cite{Edmond2013}.} In such settings, the regime may seek to deter collective action by projecting strength—convincing citizens that protest is futile. This is typically achieved by disseminating signals that suggest the regime is resilient and unlikely to fall.

In contrast, the role of an opposition leader may be to undermine this narrative. By casting doubt on the regime's strength and credibility, the opposition can shift citizens' beliefs and increase the likelihood of protest. The opposition's communication strategy thus functions as counter-propaganda, strategically aimed at encouraging collective action.

In this paper, we formalize these dynamics in a stylized global games framework following \cite{MorrisShin2003}. We consider a model with three types of agents: the regime, the opposition leader, and a continuum of citizens. At the start of the game, both the regime and the opposition observe the true strength of the regime, then simultaneously choose their levels of propaganda and counter-propaganda. Citizens subsequently receive private signals that are correlated with the regime's strength and are influenced positively by the regime's communication and negatively by the opposition's counter-message. Then the citizens decide simultaneously and independently whether to attack, based solely on these signals.

Importantly, we assume that citizens do not rationally infer the sources of bias in their signals; they interpret them at face value. This assumption captures a setting where media manipulation is not recognized or not accounted for—either due to polarization, limited attention, or information overload. Moreover, we introduce an additional source of uncertainty: a measure 
$\alpha\in[0,1]$ of partisans in the population always attacks the regime, regardless of the signal received. The actual size of this group is unknown to the regime and the opposition at the time they choose their communication strategies, which introduces strategic uncertainty.

We begin by analyzing a benchmark case where both propaganda and counter-propaganda are prohibitively costly. In this setting, we characterize the unique equilibrium as a pair of thresholds: one for the citizens' private signal and one for the regime's strength. Citizens attack the regime if their signal falls below the signal threshold, and the regime chooses to abandon power if its strength falls below the strength threshold. We show that the presence of partisans raises both thresholds, making citizen behavior more aggressive and increasing the likelihood of regime collapse (Proposition~\ref{prop:comparative}).

We then extend the model to a more realistic case in which both forms of communication are costly but feasible, with quadratic costs. In this setup, propaganda efforts are only observed in a region of intermediate regime strength. When the regime is either extremely strong or extremely weak, neither side has an incentive to expend resources on persuasion, as the outcome is essentially predetermined (Proposition~\ref{prop:dominance}).

For regimes near the point of collapse, however, communication strategies become crucial. In equilibrium, we show that the regime chooses not to engage in propaganda, as the marginal benefit of doing so vanishes at the tipping point. By contrast, the opposition leader increases counter-propaganda efforts, as the expected gain from toppling the regime outweighs the cost (Proposition~\ref{prop:main}). We derive closed-form expressions for the equilibrium levels of propaganda and counter-propaganda, as well as for the associated thresholds.

Finally, we conduct comparative statics (Proposition~\ref{prop:compstateq}) and show that the level of counter-propaganda increases with the benefit-cost ratio of the opposition leader and the precision of citizens' private signals, while it decreases with the cost of attacking. These findings highlight the strategic complementarities between information precision and political communication in regime change environments.

\subsection{Literature review}

The concept of global games was originally introduced by \cite{CarlssonvanDamme1993} in the context of coordination games with incomplete information. They showed that small perturbations in the information structure—specifically, the introduction of idiosyncratic private signals—can eliminate multiplicity and lead to a unique equilibrium. This insight was extended by \cite{MorrisShin1998} in their seminal model of currency attacks and has since been applied to a wide array of economic and political settings, including bank runs \citep{GoldsteinPauzner2005, RochetVives2004} and sovereign debt crises \citep{MorrisShin2004}.

Our paper is most closely related to two strands of the literature: regime change models with informational frictions and models of strategic communication in coordination games.

Among the former, \cite{Edmond2013} studies the role of state propaganda in suppressing collective action. In his model, rational citizens understand that the regime has incentives to bias information but may still be influenced by propaganda, especially if they cannot fully disentangle its source. However, Edmond does not model a countervailing force—such as an opposition leader—who seeks to challenge the regime's message. In contrast, we explicitly incorporate this second actor and examine how two-sided communication affects equilibrium outcomes.

\cite{deMesquita2010} also studies regime change as a coordination problem, but focuses on the signaling role of violent protest by revolutionary vanguards. In his model, small-scale violence serves as a costly signal of regime vulnerability, which may trigger broader mobilization. Our model, by contrast, emphasizes informational manipulation rather than physical confrontation. We allow both sides to engage in strategic communication, which influences citizens’ beliefs directly through perceived changes in fundamentals.

On the communication side, our work relates to \cite{Angeletosetal2006}, who examine how policy choices by the regime can signal information about underlying fundamentals. Their framework highlights how public actions influence beliefs and thereby affect the likelihood of a coordinated attack. We adopt a similar idea but introduce a second strategic communicator—namely, the opposition leader—who pulls beliefs in the opposite direction. This two-sided structure introduces novel interactions and makes the threshold behavior of citizens more complex.

Finally, our model connects to \cite{Corsettietal2004}, who study speculative attacks in currency markets involving both a continuum of small speculators and a large player. In our setting, the presence of committed partisans—whose number is unknown ex ante—plays a role analogous to that of the large player. These actors create an endogenous threshold for regime collapse, and their unobserved presence adds uncertainty for the regime and the opposition leader.

A key distinction of our work is the assumption that citizens are naive in interpreting their private signals. Unlike models in which agents rationally discount information from biased sources (e.g., \citealp{Edmond2013}), we assume citizens take their signals at face value. This provides a useful benchmark that captures environments in which media manipulation goes unrecognized. While potentially overstating the effectiveness of propaganda, this assumption allows us to isolate the role of strategic communication under informational naivete and explore the interaction between two competing narratives.

In sum, our main contribution is to integrate both regime propaganda and opposition counter-propaganda into a global-games framework with uncertainty over activist presence. We show how these competing efforts jointly shape citizens’ incentives to coordinate on regime change and offer new insights into the dynamics of political messaging and persuasion under strategic uncertainty.

\indent

The rest of the paper is organized as follows. Section \ref{sec:model} describes the model. Section \ref{sec:benchmark} considers the benchmark case in which propaganda is infinitely costly. The main analysis is contained in Section \ref{sec:main}, where we characterise an equilibrium with propaganda and counter-propaganda. Section \ref{sec:comp} discusses some comparative static results. Section \ref{sec:conclusion} concludes.

\section{\label{sec:model}The Model}
We consider a game with three types of players: the regime (R), the opposition leader (O), and a continuum of citizens indexed by $i \in [0,1]$. The core state variable is the strength of the regime, denoted by $\theta \in \mathbb{R}$, which represents the regime’s ability to resist a mass attack. This variable is drawn from an improper uniform distribution, $\theta \sim U(\mathbb{R})$,
\footnote{This is a standard simplification in the global games literature; it can be viewed as the limiting case of a normal prior $\theta \sim \mathcal{N}(0,\delta^{-1})$ as the precision $\delta \to 0$.}
and is observed by both the regime and the opposition leader.

Following this, the regime selects a level of propaganda $y \in \mathbb{R}$ at a cost $\Gamma(y)$, while the opposition simultaneously chooses a level of counter-propaganda $z \in \mathbb{R}$ at cost $\Delta(z)$. These choices influence the information environment faced by citizens.

Each citizen $i$ receives a private signal:
\begin{equation}
    x_i = \theta + y - z + \varepsilon_i,
\end{equation}
where $\varepsilon_i \sim \mathcal{N}(0, \beta^{-1})$ are i.i.d. across individuals. The parameter $\beta > 0$ reflects the precision of these private signals. 

Importantly, citizens are assumed to be \emph{naive}: they treat the distribution of their signals as if $x_i \sim \mathcal{N}(\theta, \beta^{-1})$, unaware of the endogenous influence of $y$ and $z$. This reflects a behavioral assumption consistent with limited awareness of media manipulation.

A key feature of the model is the presence of \emph{partisans}: a measure $\alpha \in [0,1]$ of citizens are committed to attacking the regime, regardless of their signal. The remaining $1 - \alpha$ citizens are strategic and choose whether to attack based on their signal $x_i$. While $\alpha$ is common knowledge among citizens at the time of their decision, it is not known to the regime or the opposition at the time they choose $y$ and $z$. From their perspective, $\alpha \sim U[0,1]$.

Let $a_i \in \{0,1\}$ denote the action of citizen $i$, with $a_i = 1$ indicating an attack. The total mass of attackers is then:
\begin{equation}
    A = \int_0^1 a_i \, di.
\end{equation}

After observing the realized attack $A$, the regime chooses whether to defend ($D = 1$) or abandon ($D = 0$) the status quo. The regime’s utility is given by:
\begin{equation}
    U_R = D(\theta - A) - \Gamma(y).
\end{equation}
That is, if the regime defends, it gains net utility $\theta - A$, reflecting its strength minus the attack size. If it abandons power, its utility is $-\Gamma(y)$.

The utility of a strategic citizen is:
\begin{equation}
    U_C = a_i (\mathbb{I}_{\{A \geq \theta\}} - c),
\end{equation}
where $c \in (0,1)$ is the cost of attacking. A citizen only benefits if the regime collapses ($A \geq \theta$). The utility from not attacking is normalized to zero.

The opposition leader receives utility:
\begin{equation}
    U_O = B \cdot \mathbb{I}_{\{A \geq \theta\}} - \Delta(z),
\end{equation}
where $B > 0$ is the gross benefit of regime collapse.

The timing of the game is as follows:
\begin{enumerate}
    \item The regime strength $\theta$ is drawn and observed by R and O.
    \item R and O simultaneously choose $y$ and $z$.
    \item The partisan share $\alpha \sim U[0,1]$ is drawn and revealed to citizens.
    \item Each citizen observes their private signal $x_i$ and chooses $a_i \in \{0,1\}$.
    \item The aggregate attack $A$ is realized.
    \item The regime chooses whether to defend or abandon.
    \item Payoffs are realized.
\end{enumerate}

This model captures a dynamic of two-sided propaganda under uncertainty about citizen commitment. Strategic uncertainty arises both from private signal noise and from the unknown number of committed attackers.

\subsection{Discussion of the Naive-Citizens Assumption}
In our model, citizens do not explicitly consider the possibility that their private signals, \( x_i = \theta + y - z + \varepsilon_i \), might be deliberately manipulated through propaganda or counter-propaganda. Instead, each citizen interprets \( x_i \) at face value, as though it were drawn from \( x_i = \theta + \varepsilon_i \). This behavioral assumption of naivete plays a pivotal role in shaping outcomes.

We adopt this assumption for two key reasons. First, it provides a tractable benchmark that clearly isolates the direct impact of communication on citizens' behavior. When individuals fail to account for strategic bias, even modest levels of propaganda or counter-propaganda can meaningfully shift perceived fundamentals. This, in turn, alters coordination thresholds in a transparent and analytically convenient way. Second, the assumption resonates with empirical realities: in many contexts, individuals consume politically charged information without carefully disentangling its origins—whether due to polarization, media saturation, or simple cognitive overload.

Even in the benchmark case where communication is prohibitively costly and thus absent, the assumption still matters. Citizens rely exclusively on their private signals, and the correlation between \( x_i \) and \( \theta \) drives the uniqueness of equilibrium. Under rational Bayesian inference, individuals would attempt to discount observed signals based on perceived incentives to manipulate them—introducing additional complexity that might obscure the core strategic interactions we wish to highlight.

If citizens were fully rational, they would realize that their signal might be systematically biased by the choices of the regime and the opposition. They would then update based not on \( x_i \), but on \( x_i - y + z \), effectively reversing the distortion. But in doing so, they would need beliefs about how \( y \) and \( z \) are chosen—beliefs that, in our model, would depend on the entire equilibrium structure. Since \( y \) and \( z \) are endogenously determined, and not simple linear functions of \( \theta \), this would make the updating process far less straightforward.

In truth, the assumption of naivete may not be unrealistic. It reflects the behavior of citizens overwhelmed by conflicting narratives, lacking the bandwidth—or the will—to critically assess every informational cue. Think of a detective trying to reconstruct the truth from a trio of witnesses. One is unbiased but vague. The other two are vocal and confident, yet clearly partial—each shouting a version that serves their cause. What is the detective to do? Without knowing which is honest and which is playing a role, he might well lean on the one voice that sounds “neutral,” even if it's just the quietest. Similarly, citizens cling to their private signal, however noisy, because it feels personal—untainted.

Ultimately, naivete captures a certain emotional truth: that in the fog of manipulation, people seek clarity not through deduction, but through trust. And when every messenger seems compromised, they may choose to trust only what they “see with their own eyes.”

\section{Benchmark Case: Infinitely Costly Lying}
\label{sec:benchmark}

We begin with the simplest scenario: both the regime and the opposition leader face infinite costs of communication. Formally, the cost functions \( \Gamma(y) \) and \( \Delta(z) \) are infinite for any nonzero values of \( y \) and \( z \), respectively. As a result, both parties choose not to engage in any communication, and we have \( y = z = 0 \).

In this setting, the strategic behavior of citizens is governed entirely by their private signals. Each non-partisan citizen follows a threshold rule: attack the regime if the signal \( x_i \) falls below a certain threshold \( x^* \). The indifference condition that defines this threshold is:
\begin{equation}
\label{eq:thresh}
    \Phi\left(\sqrt{\beta}(\theta^* - x^*)\right) = c,
\end{equation}
where \( \Phi(\cdot) \) is the standard normal cumulative distribution function, \( \theta^* \) is the regime strength below which the regime collapses, and \( c \in (0,1) \) is the cost of attacking. 

The intuition is straightforward. A citizen with signal \( x^* \) expects the regime to fall with probability \( \Phi(\sqrt{\beta}(\theta^* - x^*)) \). Since the benefit of attacking is 1 (if successful), the expected payoff from attacking equals this probability. Indifference arises when this expected benefit equals the cost of attack, \( c \).

Due to the positive correlation between \( x_i \) and \( \theta \), citizens with signals above \( x^* \) are less likely to believe the regime will fall and thus refrain from attacking. Conversely, citizens with \( x_i < x^* \) believe collapse is more likely and choose to attack. Thus, \( x^* \) partitions the population’s behavior.

Turning to the regime’s decision, recall that the regime defends the status quo (chooses \( D = 1 \)) if and only if its strength \( \theta \) exceeds the total attack level \( A \). Since a measure \( \alpha \) of partisans attacks regardless of signals, and a share \( 1 - \alpha \) of strategic citizens attacks only if \( x_i \leq x^* \), the total attack is:
\begin{equation}
    A = \alpha + (1 - \alpha) \cdot \Phi(\sqrt{\beta}(x^* - \theta)).
\end{equation}

In equilibrium, the regime is indifferent between defending and abandoning when \( \theta = A \). Plugging this into the expression above and evaluating at \( \theta = \theta^* \), we obtain:
\begin{equation}
\label{eq:theta}
    \theta^* = \alpha + (1 - \alpha) \cdot \Phi(\sqrt{\beta}(x^* - \theta^*)).
\end{equation}

Equation~\eqref{eq:theta} thus characterizes the regime strength \( \theta^* \) at which the regime is just indifferent between resisting and collapsing. Along with the citizen indifference condition~\eqref{eq:thresh}, this system determines the equilibrium thresholds \( (x^*, \theta^*) \).

Together, these two conditions describe the strategic landscape when no communication occurs: citizens base their behavior purely on private signals, while the regime anticipates the likely size of the attack based on signal distributions and the presence of partisans. Despite the simplicity of the setup, the presence of committed attackers introduces sufficient strategic uncertainty to yield rich and nontrivial comparative statics.

Since the regime’s utility is given by \( D(\theta - A) \), it will defend if and only if \( \theta > A \) and will abandon the status quo otherwise. The right-hand side of equation~\eqref{eq:theta} represents the aggregate attack faced by the regime. This includes two groups: a mass \( \alpha \) of partisans who always attack, and a mass \( 1 - \alpha \) of strategic citizens who attack only if their signal falls below the threshold \( x^* \).

By the law of large numbers,
\footnote{As is standard in the global games literature, we \emph{assume} that the law of large numbers applies to the continuum of agents. See \cite{Judd1985}.}
the fraction of strategic citizens with signals below \( x^* \), conditional on the state being \( \theta = \theta^* \), is equal to \( \Phi(\sqrt{\beta}(x^* - \theta^*)) \). Therefore, the expected total attack at state \( \theta^* \) is:
\[
    A(\theta^*; x^*, \alpha) = \alpha + (1 - \alpha) \cdot \Phi(\sqrt{\beta}(x^* - \theta^*)).
\]

This justifies the structure of equation~\eqref{eq:theta}. For a given \( \alpha \), the pair of conditions \eqref{eq:thresh}--\eqref{eq:theta} determines the unique equilibrium thresholds \( (x^*(\alpha), \theta^*(\alpha)) \).

Graphically, equation~\eqref{eq:thresh} describes a straight line in \( (x^*, \theta^*) \)-space with slope 1. It can be rearranged as:
\begin{equation}
\label{eq:line}
    \theta^* = \frac{1}{\sqrt{\beta}}\Phi^{-1}(c) + x^*.
\end{equation}

Meanwhile, equation~\eqref{eq:theta} defines an upward-sloping curve with horizontal asymptotes at \( -\infty \) and \( +\infty \). These two curves must intersect, so an equilibrium exists. We now show that the intersection is unique.

\begin{lemma}
For any \( \alpha \in (0,1) \) and \( c \in (0,1) \), the system defined by equations~\eqref{eq:thresh} and~\eqref{eq:theta} has a unique solution.
\end{lemma}

\begin{proof}
Implicitly differentiating equation~\eqref{eq:theta} with respect to \( x^* \) and \( \theta^* \), we obtain:
\[
    d\theta^* = (1 - \alpha) \cdot \phi(\sqrt{\beta}(x^* - \theta^*)) \cdot \sqrt{\beta}(dx^* - d\theta^*),
\]
which leads to:
\begin{equation}
\label{eq:slope}
    \frac{d\theta^*}{dx^*} = \frac{\sqrt{\beta}(1 - \alpha) \cdot \phi(\sqrt{\beta}(x^* - \theta^*))}{1 + \sqrt{\beta}(1 - \alpha) \cdot \phi(\sqrt{\beta}(x^* - \theta^*))}.
\end{equation}

Since \( \phi(\cdot) \), the standard normal density, is always positive, the fraction on the right is strictly less than one. Hence, the curve defined by equation~\eqref{eq:theta} has a slope strictly less than the slope of the 45-degree line defined by equation~\eqref{eq:line}. Therefore, they intersect only once.
\end{proof}

This uniqueness result mirrors the general logic in \cite{Hellwig2002}, who shows that global games admit unique equilibria when the prior is sufficiently diffuse relative to the signal precision. In our case, the prior is maximally diffuse—uniform over the real line—while the private signal has precision \( \beta > 0 \), satisfying the required condition.

We now build on this benchmark to explore how the equilibrium responds to changes in model primitives.

\subsection{Partisans and Regime Change: Comparative Statics on \( \alpha \)}

We now explore how changes in the share of committed attackers (partisans), denoted by \( \alpha \), affect equilibrium outcomes. Specifically, we ask how an increase in \( \alpha \) shifts the strategic threshold \( x^* \) used by citizens and the collapse threshold \( \theta^* \) for the regime.

\begin{proposition}
\label{prop:comparative}
An increase in the number of partisans, \( \alpha \), raises both the citizen attack threshold \( x^* \) and the regime’s collapse threshold \( \theta^* \).
\end{proposition}

\begin{proof}
Totally differentiating equations~\eqref{eq:thresh} and~\eqref{eq:theta} with respect to \( x^* \), \( \theta^* \), and \( \alpha \), and applying the Implicit Function Theorem, we obtain:
\begin{equation}
\label{eq:comparativestatics}
    \frac{d\theta^*}{d\alpha} = \frac{dx^*}{d\alpha} = 1 - \Phi(\sqrt{\beta}(x^* - \theta^*)) > 0,
\end{equation}
where the inequality follows because \( \Phi(\cdot) \in (0,1) \).
\end{proof}

The intuition is direct. A larger \( \alpha \) increases the unconditional mass of attackers, regardless of strategic behavior. For any given signal \( x_i \), the expected probability that the regime will be overthrown is now higher. This makes attacking more attractive even for citizens with relatively high signals, thereby raising the attack threshold \( x^* \). A higher \( x^* \) in turn increases the mass of strategic attackers, further amplifying the total pressure on the regime. As a result, the regime must now be stronger (i.e., have a higher \( \theta \)) to survive, which increases \( \theta^* \) as well.

\section{Costly Propaganda and Counter-Propaganda}
\label{sec:main}

We now turn to the central case of interest: strategic communication by both the regime and the opposition leader when communication is feasible but costly. For tractability, we assume quadratic cost functions:
\begin{equation}
    \Gamma(y) = \frac{y^2}{2}, \quad \Delta(z) = \frac{\psi z^2}{2},
\end{equation}
where \( \psi > 0 \) captures the relative cost-effectiveness of counter-propaganda. A smaller \( \psi \) implies that the opposition can more easily influence citizen beliefs.

This cost structure allows us to derive closed-form characterizations of optimal communication strategies.

\subsection{The Impact of Propaganda on Citizen Thresholds}

When \( y \) and \( z \) are nonzero, the equilibrium threshold \( x^* \) used by strategic citizens becomes a function of both. That is,
\[
    x^* = x^*(y, z; \alpha).
\]

Given the citizens’ naivete, they do not internalize that their signals are biased by the strategic choices of \( y \) and \( z \). Therefore, propaganda and counter-propaganda shift \( x^* \) directly and symmetrically:

\begin{lemma}
\label{lemm:direct}
The threshold \( x^*(y, z; \alpha) \) responds to changes in propaganda one-for-one:
\begin{equation}
    \frac{\partial x^*}{\partial y} = -1, \quad \frac{\partial x^*}{\partial z} = 1.
\end{equation}
\end{lemma}

That is, a unit increase in regime propaganda \( y \) is equivalent to a unit decrease in the citizen threshold \( x^* \), while a unit increase in counter-propaganda \( z \) raises \( x^* \) by one unit.

\subsection{Strategic Inactivity for Extreme Regime Strength}

Consider the realized aggregate attack when the regime has strength \( \theta \) and the partisan mass is \( \alpha \):
\begin{equation}
\label{eq:actualattack}
    A(\theta; x^*, \alpha) = \alpha + (1 - \alpha) \cdot \Phi(\sqrt{\beta}(x^* - \theta)).
\end{equation}

Given the regime’s payoff function \( U_R = D(\theta - A) - \Gamma(y) \), the regime will abandon power when \( \theta < A \). Thus, for very low values of \( \theta \), the regime is certain to fall regardless of communication. Conversely, for very high \( \theta \), the regime is sure to survive.

Since the influence of \( y \) and \( z \) cannot alter the outcome in either extreme, neither player has an incentive to invest in communication. This yields the following result:

\begin{proposition}
\label{prop:dominance}
For \( \theta \in (-\infty, 0) \cup (1, \infty) \), the optimal choices are:
\begin{equation}
    y^*(\theta) = z^*(\theta) = 0.
\end{equation}
\end{proposition}

That is, strategic communication is only used when the regime's strength lies within a critical intermediate range. Outside of this range, the outcome is predetermined, and no player has a reason to distort citizen beliefs.

\subsection{Opposition leader's problem}

For \( \theta \in [0,1] \), the success of a regime collapse depends on the realized number of partisans \( \alpha \). Since \( \Phi(\sqrt{\beta}(x^*(y,z;\alpha) - \theta)) < 1 \), the aggregate attack defined in Equation~\eqref{eq:actualattack} is strictly increasing in \( \alpha \). Therefore, for each \( \theta \in [0,1] \) and a given communication strategy \( (y,z) \), there exists a unique critical value \( \alpha^*(\theta; y, z) \) such that the regime collapses if and only if \( \alpha \geq \alpha^*(\theta; y, z) \).

Solving \( A(\theta; x^*(y,z;\alpha), \alpha) = \theta \) for \( \alpha \), we obtain:
\begin{equation}
\label{eq:explicitalpha}
    \alpha^*(\theta; y, z) = \frac{\theta - \Phi(\sqrt{\beta}(x^*(y,z;\alpha) - \theta))}{1 - \Phi(\sqrt{\beta}(x^*(y,z;\alpha) - \theta))}.
\end{equation}

This threshold satisfies \( \alpha^*(\theta; y, z) \leq 1 \) for any \( \theta \leq 1 \). Moreover, for sufficiently low \( \theta \), the numerator in Equation~\eqref{eq:explicitalpha} may become negative, implying that the regime would fall even if \( \alpha = 0 \). Let \( \underline{\theta}(y,z) \) denote the unique value satisfying
\begin{equation}
    \underline{\theta} = \Phi(\sqrt{\beta}(x^*(y,z;\alpha) - \underline{\theta})),
\end{equation}
so that \( \alpha^*(\theta; y, z) > 0 \) if and only if \( \theta > \underline{\theta}(y,z) \).

The opposition leader's expected payoff is then:
\begin{equation}
\label{eq:oppobjective}
\mathbb{E}[U_O] = -\frac{\psi z^2}{2} +
\begin{cases}
    B & \text{if } \theta \leq \underline{\theta}(y,z), \\
    B(1 - \alpha^*(\theta; y, z)) & \text{if } \underline{\theta}(y,z) < \theta \leq 1, \\
    0 & \text{if } \theta > 1.
\end{cases}
\end{equation}

This reflects that the probability of a successful regime change is \( 1 - \alpha^*(\theta; y, z) \), given that \( \alpha \sim U[0,1] \).

\subsection{Regime's problem}

The regime, observing \( \theta \) but not \( \alpha \), chooses its level of propaganda \( y \) to influence the citizen threshold \( x^*(y,z;\alpha) \). The expected payoff to the regime, integrating over the region where it survives, is
\begin{equation}
\label{eq:regimeexante}
\mathbb{E}[U_R] = \int_0^{\alpha^*(\theta; y, z)} \left(\theta - \left[\alpha + (1 - \alpha)\Phi(\sqrt{\beta}(x^*(y,z;\alpha) - \theta))\right]\right) d\alpha - \frac{y^2}{2}.
\end{equation}

We define \( \theta^* \) as the regime strength for which the regime is indifferent between maintaining and abandoning power. That is:
\begin{lemma}
The critical regime strength \( \theta^* \) solves
\begin{equation}
\label{eq:regimecritnew}
    \theta^*(y,z;\alpha) = \alpha + (1 - \alpha) \Phi(\sqrt{\beta}(x^*(y,z;\alpha) - \theta^*(y,z;\alpha))).
\end{equation}
\end{lemma}
\begin{proof}
Immediate by definition of \( \theta^* \) as the value for which the regime is exactly indifferent between resisting and conceding.
\end{proof}
 
\subsection{Citizen's problem}

The threshold \( x^* \) used by strategic citizens satisfies the same indifference condition as before:
\begin{equation}
    \Phi(\sqrt{\beta}(\theta^*(y,z;\alpha) - x^*(y,z;\alpha))) = c.
\end{equation}

Equivalently, we express this as a linear relationship:
\begin{equation}
\label{eq:citizenexante}
    \theta^*(y,z;\alpha) = \frac{1}{\sqrt{\beta}}\Phi^{-1}(c) + x^*(y,z;\alpha).
\end{equation}

\subsection{Optimality conditions}

Given the opposition leader’s problem in Equation~\eqref{eq:oppobjective}, the first-order condition with respect to \( z \) (when \( \theta > \underline{\theta}(y,0) \)) is:
\begin{equation}
    -\psi z - B\frac{\partial \alpha^*(\theta; y, z)}{\partial z} = 0.
\end{equation}

Using Equation~\eqref{eq:explicitalpha} and Lemma~\ref{lemm:direct}, we derive the optimal \( z^*(\theta) \):
\begin{equation}
\label{eq:optimalz}
    z^*(\theta) = \frac{B}{\psi}\sqrt{\beta}(1 - \theta) \cdot \frac{\phi(\sqrt{\beta}(x^*(\alpha; y, z^*) - \theta))}{[1 - \Phi(\sqrt{\beta}(x^*(\alpha; y, z^*) - \theta))]^2}.
\end{equation}

Similarly, differentiating Equation~\eqref{eq:regimeexante} with respect to \( y \), and using \( \partial x^* / \partial y = -1 \), we obtain:
\begin{equation}
\label{eq:optimalpropaganda}
    y^*(\theta) = \sqrt{\beta} \int_0^{\alpha^*(\theta; y^*, z)} (1 - \alpha) \cdot \phi(\sqrt{\beta}(x^*(\alpha; y^*, z) - \theta)) \, d\alpha.
\end{equation}

\subsection{Equilibrium Characterization}

Equilibrium consists of a quadruple \( (y^*, z^*, x^*, \theta^*) \) satisfying:
\begin{enumerate}
    \item Regime optimality: Equation~\eqref{eq:optimalpropaganda};
    \item Opposition optimality: Equation~\eqref{eq:optimalz};
    \item Regime indifference: Equation~\eqref{eq:regimecritnew};
    \item Citizen indifference: Equation~\eqref{eq:citizenexante}.
\end{enumerate}

At \( \theta = \theta^* \), we may simplify the conditions using:
\begin{equation}
    \sqrt{\beta}(x^* - \theta^*) = -\Phi^{-1}(c),
\end{equation}
which implies:
\begin{align}
    \theta^* &= \alpha^* + (1 - \alpha^*)(1 - c), \\
    y^* &= \sqrt{\beta} \phi(\Phi^{-1}(c)) \alpha^* \left(1 - \frac{\alpha^*}{2}\right), \\
    z^* &= \frac{B}{\psi} \sqrt{\beta}(1 - \alpha^*) \frac{\phi(\Phi^{-1}(c))}{c}.
\end{align}

Since \( x^* = \theta^* + y^* - z^* \), substituting gives:
\begin{equation}
    z^* - y^* = \frac{1}{\sqrt{\beta}}\Phi^{-1}(c).
\end{equation}

Solving for \( \alpha^* \), we obtain a quadratic:
\begin{equation}
    {\alpha^*}^2 - 2\left(1 + \frac{B}{\psi c}\right) \alpha^* - \frac{2\Phi^{-1}(c)}{\beta \phi(\Phi^{-1}(c))} + \frac{2B}{\psi c} = 0.
\end{equation}

In the special case where \( c = 1/2 \), we have \( \Phi^{-1}(1/2) = 0 \), simplifying to:
\begin{equation}
    {\alpha^*}^2 - 2\left(1 + \frac{2B}{\psi}\right) \alpha^* + \frac{4B}{\psi} = 0,
\end{equation}
with solution:
\begin{equation}
    \alpha^* = 1 + \frac{2B}{\psi} - \sqrt{1 + \left(\frac{2B}{\psi}\right)^2}.
\end{equation}

\begin{proposition}
\label{prop:main}
When the attack cost is \( c = 1/2 \), the equilibrium is characterized by:
\begin{align}
    \theta^* &= 1 + \frac{B}{\psi} - \frac{1}{2}\sqrt{1 + \left(\frac{2B}{\psi}\right)^2}, \\
    x^* &= \theta^*, \\
    y^* = z^* &= \frac{B}{\psi} \sqrt{\frac{2\beta}{\pi}} \left(\sqrt{1 + \left(\frac{2B}{\psi}\right)^2} - \frac{2B}{\psi}\right).
\end{align}
\end{proposition}

While we chose \( c = 1/2 \) to obtain analytic solutions, this case is also meaningful. When the noise \( \varepsilon_i \) is small, citizen beliefs are tightly clustered. Then, each individual perceives a 50\% chance of protest success—precisely matching the cost \( c = 1/2 \) for the marginal, indifferent citizen.

Thus, the equilibrium characterizes the tipping point: the regime’s last stand, the opposition’s critical push, and the knife-edge beliefs of the marginal protester.

\section{Applications and Extensions}
\label{sec:comp}

This section outlines directions in which our model can be extended and adapted to better reflect real-world political scenarios. Before turning to these extensions, we present several comparative statics results.

\subsection{Comparative Statics}

We focus on three key parameters: the opposition leader's benefit from regime collapse ($B$), the cost of counter-propaganda ($\psi$), and the precision of citizens' private signals ($\beta$). Our findings are summarized in the following proposition:

\begin{proposition}
\label{prop:compstateq}
An increase in the precision of private signals, $\beta$, raises both $y^{*}$ and $z^{*}$.

An increase in the benefit-cost ratio, $\frac{B}{\psi}$, raises $y^{*}$, $z^{*}$, $x^{*}$, and $\theta^{*}$.
\end{proposition}

The intuition behind these results is straightforward. A higher benefit from regime collapse ($B$) or a lower cost of counter-propaganda ($\psi$) encourages the opposition to increase $z$. Likewise, greater signal precision ($\beta$) makes citizens more responsive to their information, thus enhancing the effectiveness of persuasion and motivating both sides to exert greater effort.

\subsection{Introducing Dynamics}

One natural extension is to consider \emph{dynamics}.
\footnote{We thank an anonymous referee for encouraging us to explore this direction.} For instance, we may envision a setting in which the underlying fundamental ($\theta$) evolves over time, and in each period, citizens receive an additional private signal. If there is persistence in $\theta$, the model becomes a dynamic game rather than a series of static ones.

A related setup was studied by \cite{Angeletosetal2007}, who consider a case where regime strength is fixed but unobserved by the regime. The only commonly known fact at each date $t$ is that the regime has survived thus far. Their game ends upon regime collapse. Our framework could be adapted to include a second player—the opposition leader—so as to explore how recognition of manipulative messaging might provoke mass protest.

Conversely, if the regime observes unusually high levels of opposition activity, it may be induced to increase propaganda efforts in subsequent periods. Given a binding budget constraint, such spending could divert resources from security or repression, weakening the regime's long-term position.

One real-world illustration is the aftermath of the 2020 Belarusian elections, where an official result of over 80\% for Lukashenko was so implausible that it prompted widespread public protests. The manipulation, rather than deterring dissent, arguably galvanized it.

\section{Conclusion}
\label{sec:conclusion}

Over the past decade, political polarization has deepened across many societies. A key factor behind this trend is the rise of social media platforms that allow opposing factions to disseminate competing narratives. As a result, the resources allocated to propaganda—by both regimes and their opponents—have grown considerably.

\begin{sidewaystable}
\centering
\caption{Comparison with key papers in the regime-change literature}
\label{tab:comparison}
\begin{tabular}{p{3.4cm}p{3.0cm}p{3.0cm}p{3.0cm}p{3.0cm}}
\toprule
\textbf{Our paper} 
& \textbf{Edmond (2013)} 
& \textbf{De Mesquita (2010)} 
& \textbf{Angeletos \& Pavan (2006)} 
& \textbf{Corsetti et al.\ (2004)} \\
\midrule

Both the regime \emph{and} an opposition leader use communication (propaganda vs.\ counter-propaganda) 
& Only the regime uses propaganda 
& Opposition relies on \emph{violent} acts as a signal rather than informational campaigns 
& No explicit opposition role; the regime can send policy signals but does not face direct counter-propaganda 
& Focus on a currency-attack setting with a continuum of speculators plus one large speculator \\

\midrule

Regime's net payoff is $\theta - A$, with an option to defend or abandon 
& Regime's ex-post defense/abandon decision not modeled explicitly; focus on how propaganda deters attacks 
& A threshold for regime collapse is exogenous; violent vanguards push citizens to cross it 
& Regime's benefit also takes the form $\theta - A$, but with no organized opposition communicator 
& Once speculative attacks exceed a threshold, the currency peg collapses mechanically \\

\midrule

An \emph{unknown} fraction of partisans always attack; their actual share is not observed in advance 
& No partisans; all citizens are rational Bayesian updaters 
& Partisans exist, but their size is common knowledge; they engage in violence 
& No partisans, though policy actions can inadvertently reveal weaknesses 
& A single “large player” can coordinate smaller speculators; both sizes are known \\
    
\midrule

The intensity of counter-propaganda rises with the precision of private signals, as citizens do not discount manipulation 
& Even when citizens realize the regime's incentive to bias information, propaganda can still deter an attack 
& Violent signals from revolutionary vanguards can galvanize mass participation 
& Regime's own actions can backfire by signaling weakness, triggering self-fulfilling attacks 
& Large speculator orchestrates and intensifies attacks by mobilizing small speculators \\
\bottomrule
\end{tabular}
\end{sidewaystable}

This paper develops a regime change game involving three actors: a regime attempting to preserve the status quo, an opposition leader seeking to overthrow it, and a continuum of citizens who decide whether to attack. Citizens observe private signals about the regime's strength, which can be strategically distorted by both sides.

A central novelty of our approach is the presence of two strategic communicators, each attempting to shift citizens' beliefs in opposite directions. Another key feature is the inclusion of \emph{partisans}—a fraction of the population committed to attacking the regime regardless of their private signals. The true number of these partisans is not known ex ante, introducing strategic uncertainty.

We derive closed-form expressions for equilibrium propaganda levels under the assumption that the cost of attacking is $c = \frac{1}{2}$. We show that the equilibrium level of opposition counter-propaganda increases with the benefit-cost ratio $\frac{B}{\psi}$ and with the precision of private signals $\beta$, while the regime’s strength threshold $\theta^{*}$ and citizens' aggressiveness threshold $x^{*}$ also rise.

In sum, our results highlight the role of information precision and strategic uncertainty in shaping the intensity and impact of political messaging. We believe this framework opens new avenues for studying propaganda, belief formation, and regime stability in environments characterized by conflicting narratives and incomplete information.

\emph{Table \ref{tab:comparison} provides a structured overview of how our model compares with these notable contributions in the literature.}

One serious drawback of the current approach was that we have assumed that the citizens are naive, in that they do not recognise that the private signals they get ($x_i$'s) may be affected by $y$ and $z$: that is, when they form the posterior beliefs regarding $\theta$ conditional on their signals, they behave as if they have received the signals $x_i=\theta+\varepsilon_i$. One way to justify this assumption is that in a certain sense, this is how propaganda works, whereby its effect is left unnoticed by those who are affected by it.

Another limitation of our approach is that the behaviour of the partisans is left unmodeled: we simply assume that they always choose to attack. An alternative, and perhaps more realistic, way to capture the partisans' choice would be to assume that they take into account their private signals but are not influenced by propaganda.

\appendix

 In this appendix, we will show intermediate computations.
 
 \paragraph{Comments concerning the $U_R$ function: }

Since the regime makes its decision $D \in \{0,1\}$ \emph{after} observing the aggregate attack $A$, it follows that the regime’s optimal strategy is threshold-based: choose $D=1$ (defend) if and only if $\theta \geq A$, and $D=0$ (abandon) otherwise. Thus, the regime’s payoff function $U_R(\theta; A, y)$, given the propaganda level $y$, becomes a piecewise-linear function of $\theta$. Specifically,
\[
U_R(\theta; A, y) = 
\begin{cases}
-\Gamma(y), & \text{if } \theta < A,\\
\theta - A - \Gamma(y), & \text{if } \theta \geq A.
\end{cases}
\]
This function is flat for $\theta < A$ and increases linearly with slope 1 for $\theta \geq A$. The kink at $\theta = A$ reflects the regime’s discrete switch in decision.

\paragraph{Derivation of condition \eqref{eq:comparativestatics}: }

We begin by totally differentiating the citizen's indifference condition \eqref{eq:thresh}, which characterizes the signal threshold $x^*$ at which the expected benefit from attacking equals the cost:
\[
\Phi(\sqrt{\beta}(\theta^* - x^*)) = c.
\]
Using the chain rule and letting $\phi(\cdot) = \Phi'(\cdot)$ denote the standard normal density, we obtain:
\[
\sqrt{\beta}\phi(\sqrt{\beta}(\theta^* - x^*))(d\theta^* - dx^*) = 0.
\]
Since $\phi(\cdot) > 0$ and $\sqrt{\beta} > 0$, this implies:
\[
d\theta^* = dx^*.
\]
This tells us that small changes in the regime's abandonment threshold $\theta^*$ must be matched one-for-one by changes in the citizens' signal threshold $x^*$ to preserve indifference.

Next, we totally differentiate the regime’s indifference condition \eqref{eq:theta}, which determines $\theta^*$ as a function of $x^*$ and the partisan fraction $\alpha$:
\[
\theta^* = \alpha + (1 - \alpha) \Phi(\sqrt{\beta}(x^* - \theta^*)).
\]
Differentiating both sides with respect to $\alpha$, using the chain rule and applying the Leibniz rule for differentiating composite functions (justified here by continuity and differentiability of $\Phi$), we get:
\begin{align*}
d\theta^* &= d\alpha - \Phi(\sqrt{\beta}(x^* - \theta^*))\,d\alpha \\
&\quad + (1 - \alpha)\sqrt{\beta}\phi(\sqrt{\beta}(x^* - \theta^*))\, (dx^* - d\theta^*).
\end{align*}
Substituting the earlier result $dx^* = d\theta^*$ into the above yields:
\[
d\theta^* = \left(1 - \Phi(\sqrt{\beta}(x^* - \theta^*))\right)d\alpha.
\]
Dividing through by $d\alpha$ gives the expression:
\[
\frac{d\theta^*}{d\alpha} = 1 - \Phi(\sqrt{\beta}(x^* - \theta^*)).
\]
Since $d\theta^* = dx^*$, it follows that:
\[
\frac{dx^*}{d\alpha} = \frac{d\theta^*}{d\alpha} = 1 - \Phi(\sqrt{\beta}(x^* - \theta^*)),
\]
which proves condition \eqref{eq:comparativestatics}.

Note that this expression is strictly positive because $\Phi(\cdot) \in (0,1)$, hence both thresholds $x^*$ and $\theta^*$ increase in response to a rise in the partisan share $\alpha$. This matches the intuition that a greater presence of committed attackers makes marginal citizens more likely to join the attack, thereby lowering the regime's incentive to hold on.

\paragraph{Derivation of condition \eqref{eq:explicitalpha}: }

We begin by solving for the critical value $\alpha^{*}(\theta; y, z)$ such that the aggregate attack $A(\theta; x^*(y, z; \alpha), \alpha)$ equals the regime strength $\theta$. From equation \eqref{eq:actualattack}, the aggregate attack is given by:
\[
A(\theta; x^*(y,z;\alpha), \alpha) = \alpha + (1 - \alpha)\Phi\left(\sqrt{\beta}(x^*(y,z;\alpha) - \theta)\right).
\]
Setting $A(\theta; \cdot) = \theta$, we obtain:
\begin{equation}
\alpha + (1 - \alpha)\Phi\left(\sqrt{\beta}(x^* - \theta)\right) = \theta.
\end{equation}
Rewriting the left-hand side:
\[
\alpha\left[1 - \Phi\left(\sqrt{\beta}(x^* - \theta)\right)\right] + \Phi\left(\sqrt{\beta}(x^* - \theta)\right) = \theta.
\]
Subtracting the second term from both sides yields:
\[
\alpha\left[1 - \Phi\left(\sqrt{\beta}(x^* - \theta)\right)\right] = \theta - \Phi\left(\sqrt{\beta}(x^* - \theta)\right).
\]
Solving for $\alpha$, we arrive at the expression for the critical partisan threshold:
\begin{equation}
\alpha^*(\theta; y, z) = \frac{\theta - \Phi\left(\sqrt{\beta}(x^* - \theta)\right)}{1 - \Phi\left(\sqrt{\beta}(x^* - \theta)\right)},
\end{equation}
which is condition \eqref{eq:explicitalpha}.

\paragraph{Derivation of condition \eqref{eq:optimalz}: }

From \eqref{eq:explicitalpha}, the opposition leader’s success threshold $\alpha^*$ depends on the citizens’ threshold $x^*$:
\[
\alpha^* = \frac{\theta - \Phi\left(\sqrt{\beta}(x^* - \theta)\right)}{1 - \Phi\left(\sqrt{\beta}(x^* - \theta)\right)}.
\]
By Lemma \ref{lemm:direct}, we have $\partial x^* / \partial z = 1$, so:
\[
\frac{\partial \alpha^*}{\partial z} = \frac{\partial \alpha^*}{\partial x^*} \cdot \frac{\partial x^*}{\partial z} = \frac{\partial \alpha^*}{\partial x^*}.
\]
Differentiating with respect to $x^*$ using the quotient rule and the fact that $\Phi' = \phi$:
\begin{align*}
\frac{\partial \alpha^*}{\partial x^*} &= \frac{ -\sqrt{\beta}\phi\left(\sqrt{\beta}(x^* - \theta)\right)(1 - \Phi(\sqrt{\beta}(x^* - \theta))) + \sqrt{\beta}\phi\left(\sqrt{\beta}(x^* - \theta)\right)(\theta - \Phi(\sqrt{\beta}(x^* - \theta)))}{(1 - \Phi(\sqrt{\beta}(x^* - \theta)))^2} \\
&= -\sqrt{\beta}(1 - \theta)\frac{ \phi(\sqrt{\beta}(x^* - \theta)) }{ \left[1 - \Phi(\sqrt{\beta}(x^* - \theta))\right]^2 }.
\end{align*}

Substituting into the first-order condition for the opposition leader’s payoff:
\[
-\psi z - B \cdot \frac{\partial \alpha^*}{\partial z} = 0,
\]
we find:
\begin{equation}
\psi z = B \cdot \sqrt{\beta}(1 - \theta)\frac{ \phi(\sqrt{\beta}(x^* - \theta)) }{ \left[1 - \Phi(\sqrt{\beta}(x^* - \theta)) \right]^2 },
\end{equation}
so:
\begin{equation}
z^*(\theta) = \frac{B}{\psi} \cdot \sqrt{\beta}(1 - \theta)\frac{ \phi(\sqrt{\beta}(x^* - \theta)) }{ \left[1 - \Phi(\sqrt{\beta}(x^* - \theta)) \right]^2 },
\end{equation}
which is the expression in condition \eqref{eq:optimalz}.

\paragraph{Derivation of condition \eqref{eq:optimalpropaganda}: }

Recall the regime’s expected utility:
\[
\mathbb{E}[U_R] = \int_0^{\alpha^*} \left[ \theta - \left( \alpha + (1 - \alpha)\Phi(\sqrt{\beta}(x^* - \theta)) \right) \right] d\alpha - \frac{y^2}{2}.
\]
Differentiating with respect to $y$, and observing that the integrand vanishes at $\alpha = \alpha^*$ (by the definition of $\alpha^*$), we apply Leibniz’s rule:
\begin{align*}
\frac{d\mathbb{E}[U_R]}{dy} &= \underbrace{\frac{\partial \alpha^*}{\partial y} \cdot \left(\theta - A(\theta; x^*, \alpha^*)\right)}_{=0} \\
&\quad - \int_0^{\alpha^*} \left[ (1 - \alpha)\sqrt{\beta}\phi(\sqrt{\beta}(x^* - \theta)) \cdot \frac{\partial x^*}{\partial y} \right] d\alpha - y.
\end{align*}
By Lemma \ref{lemm:direct}, $\partial x^* / \partial y = -1$, so:
\[
\frac{d\mathbb{E}[U_R]}{dy} = \sqrt{\beta} \phi(\sqrt{\beta}(x^* - \theta)) \int_0^{\alpha^*} (1 - \alpha) d\alpha - y.
\]
Evaluating the integral:
\[
\int_0^{\alpha^*} (1 - \alpha) d\alpha = \alpha^* - \frac{(\alpha^*)^2}{2},
\]
we obtain:
\begin{equation}
y^* = \sqrt{\beta} \phi(\sqrt{\beta}(x^* - \theta)) \cdot \alpha^* \left(1 - \frac{\alpha^*}{2}\right),
\end{equation}
which is condition \eqref{eq:optimalpropaganda}.

\bibliographystyle{apalike}   % or abbrvnat, ecta, etc.
\bibliography{regimechange}

\end{document}